\documentclass[a4paper,11pt,conference,fleqn,final]{IEEEtran}
\usepackage[cmex10]{amsmath}
\usepackage{times,epsfig,amsfonts,amsthm, amssymb}

\newtheorem{definition}{Definition}
\newtheorem{lemma}{Lemma}
\newtheorem{theorem}{Theorem}

\ifCLASSINFOpdf
\else
\fi

\begin{document}
%
\title{Stochastic Service Guarantee Analysis Based on Time-Domain Models}

\author{\IEEEauthorblockN{Jing Xie}
\IEEEauthorblockA{Department of Telematics\\
Norwegian University of Science and Technology\\
Email: jingxie@item.ntnu.no}
\and
\IEEEauthorblockN{Yuming Jiang}
\IEEEauthorblockA{Department of Telematics \& Q2S Center\\
Norwegian University of Science and Technology\\
Email: ymjiang@ieee.org}
}


%


\maketitle

\begin{abstract}
Stochastic network calculus is a theory for stochastic service guarantee analysis of computer communication networks. In the current stochastic network calculus literature, its traffic and server models are typically defined based on the cumulative amount of traffic and cumulative amount of service respectively.~However, there are network scenarios where the applicability of such models is limited, and hence new ways of modeling traffic and service are needed to address this limitation.~This paper presents {\em time-domain} models and results for stochastic network calculus.~Particularly, we define traffic models, which are defined based on probabilistic lower-bounds on {\em cumulative packet inter-arrival time}, and server models, which are defined based on probabilistic upper-bounds on {\em cumulative packet service time}.~In addition, examples demonstrating the use of the proposed time-domain models are provided. On the basis of the proposed models, the five basic properties of stochastic network calculus are also proved, which implies broad applicability of the proposed {\em time-domain} approach.
\end{abstract}


%
\IEEEpeerreviewmaketitle

\section{Introduction}\label{Sec-intro}
Stochastic network calculus is a theory dealing with queueing systems found in computer communication networks \cite{Change:ExStoLinMax}\cite{Fidler:End2end}\\\cite{Jiang:BasicStoNet}\cite{Jiang:book}. It is particularly useful for analyzing networks where service guarantees are provided stochastically.~Such networks include wireless networks, multi-access networks and multimedia networks where applications can tolerate some certain violation of the desired performance \cite{Ferrari:Service}.  

Stochastic network calculus is based on properly defined traffic models \cite{Change:traffic}\cite{Jiang:BasicStoNet}\cite{Jiang:book}\cite{Li:traffic}\cite{David:SBB}\cite{Yaron:traffic} and server models \cite{Jiang:BasicStoNet}\cite{Jiang:book}. In the existing models of stochastic network calculus, an arrival process and a service process are typically modeled by some stochastic arrival curve, which probabilistically upper-bounds the {\em cumulative amount of arrival},~and respectively by some stochastic service curve, which probabilistically lower-bounds the {\em cumulative amount of service}.~In this paper, we call such models {\em space-domain} models.~Based on the {\em space-domain} traffic and server models, a lot of results have been derived for stochastic network calculus.~Among the others, the most fundamental ones are the five basic properties \cite{Jiang:BasicStoNet} \cite{Jiang:book}:~(P.1)~{\em Service Guarantees} including delay bound and backlog bound;~(P.2)~{\em Output Characterization};~(P.3)~{\em Concatenation Property}; (P.4) {\em Leftover Service}; (P.5) {\em Superposition Property}. Examples demonstrating the necessity of having these basic properties and their use can be found \cite{Jiang:BasicStoNet} \cite{Jiang:book}.

Nevertheless, there are still many open research challenges for stochastic network calculus, and a critical one is {\em time-domain modeling and analysis} \cite{Jiang:book}.~Time-domain modeling for service guarantee analysis has its root from the deterministic Guaranteed Rate (GR) server model \cite{Goyal:GR}, where service guarantee is captured by comparing with a (deterministic) virtual time function in the time-domain.~This time-domain model has been extended to design aggregate-scheduling networks to support per-flow (deterministic) service guarantees \cite{Cobb:Qua}\cite{jiang:Time}, while few such results are available from space-domain models. Other network scenarios where time-domain modeling may be preferable include wireless networks and multi-access networks. 

In wireless networks, the varying link condition may cause failed transmission when the link is in \lq bad\rq~condition.~The sender may hold until the link condition becomes \lq good\rq~or re-transmit. For such cases, it is difficult to directly find the stochastic service curve in the space-domain because we need to characterize the stochastic nature of the impaired service caused by the \lq bad\rq~link condition.~A possible way is that we use an impairment process \cite{Jiang:BasicStoNet} to characterize the impaired service.~However, how to define and find the impairment process arises another difficulty. Even though we can define an impairment process, we may first convert the impairment process into some existing stochastic network calculus models, and then further analyze the performance bounds. The obtained performance bounds may become loose because of such conversion. If we characterize the serivce process in the time-domain, we can use random variables to represent the time intervals when the link is in \lq bad\rq~condition. Analyzing the stochastic nature of such random variables would be easier. In addition, this way can avoid the difference introduced by the intermediate conversion. 

In contention-based multi-access networks, backoff schemes are often employed to reduce collision occuring. Because the backoff process is characterized by backoff windows which may vary with the different backoff stages, it is quite cumbersome for a space-domain server model to characterize the service process with the consideration of the backoff process. This also prompts the possibility of characterzing the service process in the time-domain. Having said this, however, how to define a stochastic version of the virtual time function and how to perform the corresponding analysis are yet open \cite{Jiang:book}.

The objective of this paper is to define traffic models and server models in the {\em time-domain} and derive the corresponding five basic properties for stochastic network calculus. Particularly, we define traffic models that are based on probabilistic lower bounds on {\em cumulative packet inter-arrival time}. Also, we define server models that are based on some virtual time function and probabilistic upper bounds on {\em cumulative packet service time}. In addition, we establish relationships among the proposed time-domain models, and the mappings between the proposed time-domain models and the existing space-domain models. Furthermore, we prove the five basic properties based on the proposed time-domain models.

The remainder is structured as follows. Sec.~\ref{Sec-Background} introduces the mathematical background and fundamental space-domain models and relevant results of stochastic network calculus. In Sec.~\ref{Section-Model}, we first introduce the time-domain deterministic traffic and server models, and then extend them to stochastic versions.~In addition, the relationships among them as well as with some existing space-domain models are established.~Sec.~\ref{Sec-Property} explores the five basic properties. Sec.~\ref{Sec-Conclusion} summarizes the work.

\section{Notation and Background}\label{Sec-Background}
To ease expression, we assume networks with \textbf{fixed unit length}\footnote{The results can also be extended to networks with variable-length packets while the expression and results will be more complicated.} packets.~By convention, we assume that a packet is considered to be received by a network element when and only when its last bit has arrived to the network element, and a packet is considered out of a network element when and only when its last bit has been transmitted by the network element.~A packet can be served only when its last bit has arrived.~All queues are assumed to be empty at time $0$. Packets within a flow are served in the first-in-first-out (FIFO) order. 

\subsection{Notation}
Let $p^n$, $r(n)$, $a(n)$ and $d(n)$ $(n=0,1,2,...)$ denote the $n^{th}$ packet of a flow, its allocated service rate, its arrival time and its departure time, respectively.~Let $\mathcal{A}(t)$ and $\mathcal{A}^*(t)$ respectively denote the number of cumulative arrival packets and the number of cumulative departure packets by time $t$. By convention, we assume $a(0)=0$, $d(0)=0$, $\mathcal{A}(0)=0$ and $\mathcal{A}^*(0)=0$. For any $0\leq s\leq t$, we denote $\mathcal{A}(s,t)\equiv\mathcal{A}(t)-\mathcal{A}(s)$ and $\mathcal{A}^*(s,t)\equiv\mathcal{A}^*(t)-\mathcal{A}^*(s)$. 

In this paper, $a(n)$ and $\mathcal{A}(t)$ will be used to represent an arrival process interchangeably. A departure process will be represented by $d(n)$ and $\mathcal{A}^*(t)$ interchangeably.

The following function sets are often used in this paper. Specifically, we use $\mathcal{G}$ to denote the set of non-negative wide-sense increasing functions as follows:
\begin{displaymath}
\small{\mathcal{G} = \{g(\cdot): \forall 0\leq x\leq y, 0 \leq g(x) \leq g(y)\}}
\end{displaymath}
We denote by $\bar{\mathcal{G}}$ the set of non-negative wide-sense decreasing functions:
\begin{displaymath}
\small{\bar{\mathcal{G}} = \{g(\cdot): \forall 0\leq x\leq y, 0 \leq g(y) \leq g(x)\}}
\end{displaymath}
Let $\bar{\mathcal{F}}$ denote the set of functions in $\bar{\mathcal{G}}$, where for each function $f(\cdot)\in\bar{\mathcal{F}}$, its nth-fold integration, denoted by $f^{(n)}(x)\equiv\big(\int_{x}^{\infty}dy\big)^nf(y)$, is bounded for $\forall x\geq 0$ and still belongs to $\bar{\mathcal{F}}$ for $\forall n\geq 0$, or 
\begin{displaymath}
\bar{\mathcal{F}} = \big\{f(\cdot):\forall n\geq 0, \big(\int_{x}^{\infty}dy\big)^nf(y)\in\bar{\mathcal{F}}\big\}. 
\end{displaymath}
For ease of exposition, we adopt
\begin{displaymath}
[x]^{+} \equiv \max[0,x] ~~{and}~~ [x]_1 \equiv \min[1,x], 
\end{displaymath}
and assume that for any bounding function $f(x)$, $f(x)=1$ for $\forall x<0$. 

\subsection{Max-plus and Min-plus Algebra Basics}
An essential idea of (stochastic) network calculus is to use alternate algebras particularly the min-plus algebra and max-plus algebra \cite{Boudec:calculus} to transform complex non-linear network systems into analytically tractable linear systems \cite{Jiang:book}. To the best of our knowledge, the existing models and results of stochastic network calculus are mainly under the {\em space-domain} and based on min-plus algebra that has basic operations particularly suitable for characterizing cumulative arrival and cumulative service.~For characterizing arrival and service processes in the \emph{time-domain}, interestingly, the max-plus algebra has basic operations that well suit the need. 
 
In this paper, the following \emph{max-plus} and \emph{min-plus} operations are often used:
\begin{itemize}
\item \emph{Max-Plus Convolution} of $g_1$ and $g_2$ is 
\begin{displaymath}
(g_1 \bar{\otimes} g_2)(x) = \sup_{0 \leq y \leq x}\{g_1(y) + g_2(x - y)\}
\end{displaymath}
\item \emph{Max-Plus Deconvolution} of $g_1$ and $g_2$ is 
\begin{displaymath}
(g_1 \bar{\oslash} g_2)(x) = \inf_{y\geq 0}\{g_1(x + y) - g_2(y)\}
\end{displaymath}
\item \emph{Min-Plus Convolution} of $g_1$ and $g_2$ is 
\begin{displaymath}
(g_1\otimes g_2)(x) = \inf_{0 \leq y \leq x}\{g_1(y) + g_2(x - y)\}
\end{displaymath}
\item \emph{Min-Plus Deconvolution} of $g_1$ and $g_2$ is 
\begin{displaymath}
(g_1\oslash g_2)(x) = \sup_{y \geq 0}\{g_1(x + y) - g_2(y)\}
\end{displaymath}
\end{itemize}
In this paper, when applying \emph{supremum} and \emph{infimum}, they may be interpreted as \emph{maximum} and \emph{minimum} whenever appropriate, respectively. 

\subsection{Preliminaries}
The following lemma is often used for later analysis and thus listed: 
\begin{lemma}\label{lemma1}
For the sum of a collection of random variables $Z = \sum_{i=1}^n X_i$, no matter whether they are independent or not, there holds for the complementary cumulative distribution function (CCDF) of $Z$: (See Lemma 1.5 in \cite{Jiang:book})
\begin{equation}
\small{\bar{F}_Z(z) \leq \bar{F}_{X_1}\otimes\cdot\cdot\cdot\otimes\bar{F}_{X_n}(z)}
\label{eq:CCDF}
\end{equation}
where $\bar{F}_Z = P\{Z > z\}$, $-\infty < z \leq \infty$.
\end{lemma}

For later analysis, we need some transformation between the number of cumulative arrival packets by time $t$, i.e., $\mathcal{A}(t)$, and the time of a packet arriving to the system, i.e., $a(n)$. 

If $\mathcal{A}(t)$ is upper-bounded with respect to some function $\alpha(t)\in\mathcal{G}$, we have the following lemma.
\begin{lemma}\label{spacedomain}
For function $\alpha(t)\in\mathcal{G}$, there holds:
\begin{enumerate}
\item the following statements are equivalent: 
\begin{enumerate}
\item $\forall 0\leq s\leq t$, $\mathcal{A}(s,t)\leq \alpha(t-s)+x$ for $\forall x\geq 0$;
\item $\forall t\geq 0$, $\mathcal{A}(t)\leq \mathcal{A}\otimes\alpha(t)+x$ for $\forall x\geq0$;
\end{enumerate}
\item if $\forall t,x\geq 0$, $\mathcal{A}(t)\leq \mathcal{A}\otimes\alpha(t)+x$ holds, then we have $a(n) \geq a\bar{\otimes}\lambda(n)-y$, where $\lambda(n)\in\mathcal{G}$ is the \emph{inverse} function of $\alpha(t)$ and defined as follows
\begin{equation}
\lambda(n) = \inf\{\tau: \alpha(\tau)\geq n\}
\label{eq:inversealpha} 
\end{equation}
and
\begin{equation}
y= \sup_{k\geq 0}[\lambda(k)-\lambda(k-x)].
\label{eq:yinversefunction}
\end{equation}
\end{enumerate}
\end{lemma}

\begin{proof}
(1) For $(a)\rightarrow(b)$, from the condition, we obtain $\mathcal{A}(s,t)-\alpha(t-s)-x\leq 0$ for $\forall~0\leq s\leq t$. Then, there holds
\begin{displaymath}
~~~~~~~~\sup_{0\leq s\leq t}[\mathcal{A}(s,t)-\alpha(t-s)-x]\leq 0
\end{displaymath}
which implies 
\begin{displaymath}
~~~~~~~~\mathcal{A}(t)-\inf_{0\leq s\leq t}[\mathcal{A}(s)+\alpha(t-s)]-x\leq 0.
\end{displaymath}
Thus, we conclude $\mathcal{A}(t)\leq \mathcal{A}\otimes\alpha(t)+x$ for $\forall t,x\geq 0$.

For $(b)\rightarrow(a)$, from the condition, we have 
\begin{displaymath}
~~~~~~~~\mathcal{A}(t)-\inf_{0\leq s\leq t}[\mathcal{A}(s)+\alpha(t-s)]-x\leq 0
\end{displaymath}
 which implies 
\begin{displaymath}
~~~~~~~~\sup_{0\leq s\leq t}[\mathcal{A}(s,t)-\alpha(t-s)-x]\leq 0.
\end{displaymath}
Then there must be $\mathcal{A}(s,t)-\alpha(t-s)-x\leq 0$ for $\forall 0\leq s\leq t$. Thus, $\mathcal{A}(s,t)\leq \alpha(t-s)+x$ holds for $\forall 0\leq s\leq t$ and $\forall x\geq 0$.

(2) From (1), we know that $\mathcal{A}(t)\leq \mathcal{A}\otimes\alpha(t)+x$ is equivalent to $\mathcal{A}(s,t)\leq \alpha(t-s)+x$ for $\forall 0\leq s\leq t$ and $\forall x\geq 0$. Then for $\forall 0\leq m \leq n$, we have
\begin{displaymath}
~~~~~~~~\mathcal{A}\big(a(m),a_{+}(n)\big)\leq\alpha\big(a_{+}(n)-a(m)\big)+x
\end{displaymath}
where $a_{+}(n)=a(n)+\epsilon$ with $\epsilon\rightarrow 0$. We also know 
\begin{displaymath}
n-m \leq \mathcal{A}\big(a(m),a_{+}(n)\big)\leq \alpha\big(a_{+}(n)-a(m)\big)+x
\end{displaymath}
Taking the inverse function of $\alpha\big(a_{+}(n)-a(m)\big)$ yields 
\begin{displaymath}
a_{+}(n)-a(m) \geq \lambda(n-m-x)~~~~~~~~~~~~~~~~~~~~~~~~~~~~~~~~~~~~~~~~~~~~~~
\end{displaymath}
\begin{displaymath}
~~~~= \lambda(n-m) - [\lambda(n-m)-\lambda(n-m-x)]
\end{displaymath}
\begin{equation}
~~~~~~~~\geq \lambda(n-m) - \sup_{n-m\geq 0}[\lambda(n-m)-\lambda(n-m-x)]
\label{eq:proofstep1}
\end{equation}
Let $k=n-m$. Eq.(\ref{eq:proofstep1}) can be written as 
\begin{displaymath}
a_{+}(n)-a(m)\geq\lambda(n-m)-\sup_{k\geq 0}[\lambda(k)-\lambda(k-x)]
\end{displaymath}
from which we obtain 
\begin{equation}
~~~~~~~~a(n)\geq a(m)+ \lambda(n-m) - y
\label{eq:proofstep2}
\end{equation}
because $\epsilon\rightarrow 0$ and $y=\sup_{k\geq 0}[\lambda(k)-\lambda(k-x)]$. 
Since Eq.(\ref{eq:proofstep2}) holds for $\forall 0\leq m \leq n$, we have
\begin{displaymath}
a(n) \geq \sup_{0\leq m\leq n}[a(m)+\lambda(n-m)-y]=a\bar{\otimes}\lambda(n)-y.
\end{displaymath}
\end{proof}

\textbf{Example 1.} Suppose the number of cumulative arrival packets of a flow, $\mathcal{A}(t)$, is upper-bounded by $\alpha(t)+x$ for $t\geq 0$, where $\alpha(t)=\rho\cdot t+\sigma$. Let $n\equiv\alpha(t)$. We can get the inverse function of $\alpha(t)$, $\lambda(n)=\frac{(n-\sigma)^+}{\rho}$. We can use Eq.(\ref{eq:yinversefunction}) to get $y$. For $\forall k\geq 0$, we have
\begin{displaymath}
~~~~~~~~y=\sup_{k\geq 0}\Big\{\frac{(k-\sigma)^+}{\rho}-\frac{(k-\sigma-x)^+}{\rho}\Big\}
\end{displaymath}
\begin{displaymath}
~~~~~~~~~~=\begin{cases}
\frac{x}{\rho}~~~~~~~~ \text{k}\geq\sigma+x,\\
< \frac{x}{\rho}~~~~~\sigma\leq\text{k}<\sigma+x,\\
0~~~~~~~~\text{k}<\sigma.
\end{cases}
\end{displaymath}
from which we get $y=\frac{x}{\rho}$. Then, we know that for any packet, its arrival time satisfies
\begin{displaymath}
a(n)\geq \sup_{0\leq m\leq n}\big[a(m)+\frac{(n-m-\sigma)^+}{\rho}\big]-\frac{x}{\rho}.
\end{displaymath} 

If $a(n)$ is lower-bounded with respect to some function $\lambda(n)\in\mathcal{G}$ , we have the following lemma. 
\begin{lemma}\label{timedomainlemma}
For function $\lambda(n)\in\mathcal{G}$, there holds:
\begin{enumerate}
\item the following statements are equivalent:
\begin{enumerate} 
\item $\forall 0\leq m \leq n$, $a(n)-a(m)\geq\lambda(n-m)-y$ for $\forall y\geq 0$;  
\item $\forall n\geq 0$, $a(n) \geq a\bar{\otimes}\lambda(n)-y$ for $\forall y\geq 0$;
\end{enumerate}
\item if $\forall n,y\geq 0$, $a(n) \geq a\bar{\otimes}\lambda(n)-y$ holds, then we have $\mathcal{A}(t)\leq \mathcal{A}\otimes\alpha(t)+x$, where $\alpha(t)\in\mathcal{G}$ is the inverse function of $\lambda(n)$ and defined as follows
\begin{equation}
\alpha(t) = sup\{k: \lambda(k)\leq t\}
\end{equation}
and 
\begin{equation}
x = \sup_{\tau\geq 0}[\alpha(\tau+y)-\alpha(\tau)+1].
\label{eq:xvalue}
\end{equation}
\end{enumerate}
\end{lemma}

\begin{proof}
(1) The $(a)\rightarrow(b)$ part has been proved in Lemma \ref{spacedomain}(2). We only prove the $(b)\rightarrow(a)$ part. From the condition, we have 
\begin{displaymath}
a(n)-\sup_{0\leq m\leq n}\{a(m)+\lambda(n-m)\}+y\geq 0
\end{displaymath}
which implies 
\begin{displaymath}
\inf_{0\leq m \leq n}\{a(n)-a(m)-\lambda(n-m)\}+y\geq 0.
\end{displaymath}
Thus there holds $a(n)-a(m)\geq\lambda(n-m)-y$ for $\forall 0\leq m \leq n$ and $\forall y\geq 0$. 

(2) For $\forall 0\leq s\leq t$, we can find $m,n\geq 0$ according to the following functions
\begin{displaymath}
~~~~~~\mathcal{A}(t) = n =\sup\{k:a(k)\leq t\}
\end{displaymath}
\begin{displaymath}
~~~~~~\mathcal{A}(s) = m =\sup\{k:a(k)\leq s\}.
\end{displaymath} 
Thus, we have $\mathcal{A}(s,t)=n-m$ and $a(n)-a(m+1)\leq t-s$. From (1), we know that $a(n)\geq a\bar{\otimes}\lambda(n)-y$ is equivalent to $a(n)-a(m)\geq\lambda(n-m)-y$. Then we have
\begin{displaymath}
~~~~~~t-s \geq a(n)-a(m+1)\geq \lambda(n-m-1)-y.
\end{displaymath}
Taking the inverse function of $\lambda(n-m-1)$ yields
\begin{displaymath}
~~~~~~n-m-1 \leq \alpha(t-s+y)
\end{displaymath}
Because $\mathcal{A}(s,t)=n-m$, we have
\begin{displaymath}
\mathcal{A}(s,t)\leq \alpha(t-s+y)+1~~~~~~~~~~~~~~~~~~~~~~~~~~~~~~~~~~~~~~~~~~~~~
\end{displaymath}
\begin{displaymath}
~~~~~~= \alpha(t-s)+[\alpha(t-s+y)-\alpha(t-s)+1]
\end{displaymath}
\begin{displaymath}
~~~~~~\leq \alpha(t-s)+\sup_{t-s\geq 0}[\alpha(t-s+y)-\alpha(t-s)+1]
\end{displaymath}
Let $\tau=t-s$. The above inequality is written as
\begin{displaymath}
\mathcal{A}(s,t)\leq \alpha(t-s)+\sup_{\tau\geq 0}[\alpha(\tau+y)-\alpha(\tau)+1]
\end{displaymath}
\begin{displaymath}
~~~~~~~~= \alpha(t-s)+x
\end{displaymath}
Since $\mathcal{A}(s,t)-\alpha(t-s)-x\leq 0$ holds for $\forall 0\leq s\leq t$, we have
\begin{displaymath}
~~~~~~\sup_{0\leq s\leq t}[\mathcal{A}(s,t)-\alpha(t-s)-x]\leq 0
\end{displaymath}
from which we further obtain
\begin{displaymath}
~~~~~~\mathcal{A}(t)-\inf_{0\leq s\leq t}[\mathcal{A}(s)+\alpha(t-s)]-x\leq 0.
\end{displaymath}
We then conclude $\mathcal{A}(t)\leq\mathcal{A}\otimes\alpha(t)+x$.
\end{proof}

\subsection{Related Space-domain Results}
This sub-section reviews some related {\em space-domain} results under min-plus algebra \cite{Jiang:book}.~It is worth highlighting that the following results are for discrete time systems with unit discretization step. 

The \emph{virtual-backlog-centric} (v.b.c) stochastic arrival curve model \cite{Jiang:TrafficModel} is defined based on a probabilistic upper-bound on cumulative arrival. To ease later analysis, the definition of v.b.c stochastic arrival curve model presented in this paper is based on the number of cumulative arrival packets while not the amount of cumulative arrival (in bits) which has been widely used in the network calculus literature.  

The v.b.c stochastic arrival curve model explores the \emph{virtual backlog property} of deterministic arrival curve, which is that the queue length of a virtual single server queue (SSQ) fed with the same flow with a deterministic arrival curve is upper-bounded. 

For a flow having arrival curve $\alpha(t)$, we construct a virtual SSQ system fed with the same flow. The SSQ system has infinite buffer space and the buffer is initially empty. Suppose the virtual SSQ system provides service $\alpha(t)$ to the flow for all $t\geq 0$. Then the unfinished work or backlog in the virtual SSQ system by time $t$ is $\mathcal{B}(t)=\mathcal{A}(t)-\mathcal{A}^*(t)$. The Lindely equation can be used to derive $\mathcal{B}(t)$, which is 
\begin{equation}
\mathcal{B}(t) = \max\{0,\mathcal{B}(t-1)+\mathcal{A}(t-1,t)-\alpha(t-t+1)\}
\label{eq:LindelEq}
\end{equation}
Eq.(\ref{eq:LindelEq}) means that the amount of traffic backlogged in the system by time $t$ equals the amount of traffic backlogged by time $t-1$ plus the amount of traffic having arrived between $t-1$ and $t$ minus the amount of traffic having been served between $t-1$ and $t$. By applying Eq.(\ref{eq:LindelEq}) iteratively to its right-hand side, it becomes
\begin{equation}
\mathcal{B}(t) =\sup_{0\leq s\leq t}[A(s,t)-\alpha(t-s)].
\label{eq:backlog}
\end{equation}
 If the flow is constrained by arrival curve $\alpha(t)+x$ for all $t\geq 0$, it follows from Eq.(\ref{eq:backlog}) that the system backlog is also upper-bounded by $x$. The v.b.c stochastic arrival curve is defined based on the \emph{virtual backlog property}. 
\begin{definition}\label{vbcArrCurve}
\textbf{(v.b.c Stochastic Arrival Curve).}

A flow is said to have a virtual-backlog-centric (v.b.c) sto-\\chastic arrival curve $\alpha(t)\in\mathcal{G}$ with bounding function $f(x)\in\bar{\mathcal{G}}$, denoted by $\mathcal{A}\sim_{vb}\langle\alpha,f\rangle$, if for all $t\geq 0$ and all $x\geq 0$, there holds
\begin{equation}
P\big\{\sup_{0\leq s\leq t}[\mathcal{A}(s,t) - \alpha(t-s)] > x\big\}\leq f(x).
\label{eq:vbcArrCurve}
\end{equation}
\end{definition}
Eq.(\ref{eq:vbcArrCurve}) can also be written as follows:
\begin{equation}
P\big\{\mathcal{A}(t) > \mathcal{A}\otimes\alpha(t) + x\big\}\leq f(x).
\end{equation}

Based on the existing space-domain traffic and server models, a lot of results have been derived for stochastic network calculus which include the five basic properties \cite{Jiang:book} as introduced in Sec.~\ref{Sec-intro}.~In this paper, the following result is specifically made use of in later analysis and hence listed:
\begin{lemma}\label{superpositionMinPlus}
\textbf{(Superposition Property).} Consider $N$ flows with arrival processes $\mathcal{A}_i(t)$, i=1,...,N, respectively.~Let $\mathcal{A}(t)$ denote the aggregate arrival process.~If $\forall i$, $\mathcal{A}_i\backsim_{vb}\langle\alpha_i,f_i\rangle$, then $\mathcal{A}\backsim_{vb} \langle\alpha,f\rangle$ with $\alpha(t) = \sum_{i=1}^N\alpha_i(t)$, and $f(x) = f_1\otimes\cdot\cdot\cdot\otimes f_N(x)$.
\end{lemma}

\section{Time-Domain Models}\label{Section-Model}
This section reviews the deterministic arrival curve and the deterministic service curve models defined in the time-domain. We generalize the deterministic models and define {\em time-domain} stochastic arrival curve and stochastic service curve models.
   
\subsection{Deterministic Arrival Curve}
Consider a flow of which packets arrive to a system at time $a(n)$. In order to deterministically guarantee a certain level of quality of service (QoS) to this flow, the traffic sent by this flow must be constrained. The deterministic network calculus traffic model in the time-domain characterizes packet inter-arrival time using a lower-bound function, called arrival curve in this paper and defined as follows \cite{Chang:MaxPlus}:
\begin{definition}\label{DetArriCurve}
\textbf{(Arrival Curve).} A flow is said to have a (deterministic) arrival curve $\lambda(n)\in\mathcal{G}$, if its arrival process $a(n)$ satisfies, for all $0\leq m\leq n$, 
\begin{equation}
~~~~~~~~a(n) - a(m) \geq \lambda(n-m).
\end{equation}
\end{definition}

The arrival curve model has the following triplicity principle which will be used as the basis in defining the stochastic arrival curve models in the subsequent subsections.  
\begin{lemma}\label{triplicity}
The following statements are equivalent:
\begin{enumerate}
\item $\forall 0\leq m\leq n$, $a(n)-a(m)\geq\big[\lambda(n-m)-x\big]^+$ for $\forall x\geq 0$;
\item $\forall$$n\geq 0$, $\sup_{0\leq m\leq n}\Big\{\big[\lambda(n-m)-x\big]^+ - [a(n) - a(m)]\Big\}\leq 0$ for $\forall x\geq 0$;
\item $\forall n\geq 0$, $\sup_{0\leq m\leq n}\sup_{0\leq q\leq m}\Big\{\big[\lambda(m-q)-x\big]^+ - [a(m) - a(q)]\Big\}\leq 0$ for $\forall x\geq 0$,
\end{enumerate}
where $\lambda\in\mathcal{G}$.
\end{lemma} 
\begin{proof}
It is trivially true that 
\begin{displaymath}
\lambda(n-m) - [a(n)-a(m)]\leq\sup_{0\leq m\leq n}\{\lambda(n-m)-[a(n)-a(m)]\}
\end{displaymath}
from which, (2) implies (1). In addition 
\begin{displaymath}
\sup_{0\leq m\leq n}\{\lambda(n-m) - [a(n) - a(m)]\}~~~~~~~
\end{displaymath}
\begin{displaymath}
\leq \sup_{0\leq m\leq n}\sup_{m\leq k\leq n}\big\{\lambda(k-m) - [a(k) - a(m)]\big\}
\end{displaymath}
\begin{displaymath}
= \sup_{0\leq k\leq n}\sup_{0\leq m\leq k}\big\{\lambda(k-m) - [a(k) - a(m)]\big\}
\end{displaymath}
\begin{displaymath}
= \sup_{0\leq m\leq n}\sup_{0\leq q\leq m}\big\{\lambda(m-q) - [a(m) - a(q)]\big\}
\end{displaymath}
with which, (3) implies (2).

For (1)$\rightarrow$(2), it holds since $a(n)-a(m)\geq\lambda(n-m)-x$ for $\forall 0\leq m\leq n$. For (2)$\rightarrow$(3), 
\begin{displaymath}
\sup_{0\leq m\leq n}\sup_{0\leq q\leq m}\{\lambda(m-q)-[a(m)-a(q)]\}\leq\sup_{0\leq m\leq n}[x] = x.
\end{displaymath} 
Thus (1), (2) and (3) are equivalent. 
\end{proof}
From Definition \ref{DetArriCurve}, the right-hand side of $a(n)-a(m)\geq\lambda(n-m)-x$ in Lemma \ref{triplicity}.(1) defines an arrival curve $\lambda(n-m)-x$. In addition, we can construct a virtual single server queue (SSQ) system that is initially empty, fed with the same traffic flow, and has a service curve $\lambda$ which makes $d(n)\leq a\bar{\otimes}\lambda(n)$ (see Definition \ref{DetServerModel}).~Then, the delay in the virtual SSQ system is upper-bounded by $d(n)-a(n)\leq\sup_{0\leq m \leq n}[\lambda(n-m)-(a(n)-a(m))]\leq x$, and the maximum system delay for the first $n$ packets is upper-bounded by 
\begin{displaymath}
\sup_{0\leq m \leq n}\{d(m)-a(m)\}
\end{displaymath}
\begin{displaymath}
\leq\sup_{0\leq m\leq n}\sup_{0\leq q\leq m}\{\lambda(m-q)-[a(m)-a(q)]\}\leq x.
\end{displaymath}

\textbf{Example 2}. The Generic Cell Rate Algorithm (GCRA) \cite{ITU:GCRA} with parameter $(T,\tau)$ is a parallel algorithm to the Leaky Bucket algorithm and has been used in fixed-length packet networks such as Asynchronous Transfer Mode (ATM) networks. The GCRA measures cell rate at a specified time scale and assumes that cells will have a minimum interval between them. Here, $T$ denotes the assumed minimum interval between cells and $\tau$ denotes the maximum acceptable excursion that quantifies how early cells may arrive with respect to $T$.~It can be verified that if a flow is GCRA$(T,\tau)$-constrained, it has an arrival curve
\begin{displaymath}
~~~~~~~~~~~~~~\lambda(n) = \big(T\cdot n - \tau\big)^+.
\end{displaymath}

\subsection{Inter-arrival-time Stochastic Arrival Curve}\label{Sec-itarrcur}
Lemma \ref{triplicity}.(1) defines a deterministic arrival curve $\lambda(n)-x$ which lower-bounds the inter-arrival time between any two packets.~Based on this, we define its probabilistic counterpart as follows:
\begin{definition}\label{itarrcurve}
\textbf{(i.a.t Stochastic Arrival Curve).} A flow is said to have an inter-arrival-time (i.a.t) stochastic arrival curve $\lambda\in\mathcal{G}$ with bounding function $h\in\bar{\mathcal{G}}$, denoted by $a(n)\sim_{it}\langle\lambda,h\rangle$, if for all $0\leq m\leq n$ and all $x\geq 0$, there holds
\begin{equation}
P\Big\{\lambda(n-m) - [a(n) - a(m)] > x\Big\}\leq h(x).
\end{equation}
\end{definition}

\textbf{Example 3.} Consider a flow with fixed unit packet size. Suppose its packet inter-arrival times follow an exponential distribution with mean $\frac{1}{\rho}$. Then, the packet arrival time has an Erlang distribution with parameter $(n,\rho)$ \cite{handbook:erlangdist}. And, for any two packets $p^m$ and $p^n$, their inter-arrival time $a(n)-a(m)$ satisfies, for $\forall x\geq 0$,
\begin{displaymath}
P\Big\{\frac{1}{\rho}(n-m)-[a(n)-a(m)] > x\Big\}~~~~~~~~~~~~~~~
\end{displaymath}   
\begin{displaymath}
\leq 1-\sum_{k=0}^{n-m-1}\frac{e^{-\rho y}(\rho y)^k}{k!}-\rho\frac{e^{-\rho y}(\rho y)^{n-m-1}}{(n-m-1)!} 
\end{displaymath} 
where $y=\frac{1}{\rho}(n-m)-x$. 

The i.a.t stochastic arrival curve is intuitively simple, but it has limited use if no additional constraint is enforced. Let us consider a simple example to understand this problem. Consider a single node with constant per packet service time $T$ and its input flow $F$ satisfying $a(n)\sim_{it}\langle\tau\cdot n,h\rangle$ where $\tau\geq T$. Suppose we are interested in the delay $D(n)$, where, by definition, $D(n)=d(n)-a(n)$. Because the node has constant per packet service time $T$, it has a (deterministic) service curve $T\cdot n$ which implies $d(n)=\sup_{0\leq m\leq n}[a(m)+T\cdot(n-m)]$. Then we have
\begin{displaymath}
D(n)=\sup_{0\leq m\leq n}\big\{a(m)+T\cdot(n-m)\big\}-a(n)
\end{displaymath}
\begin{displaymath}
~~~~~~= \sup_{0\leq m\leq n}\big\{a(m)+T\cdot(n-m)-a(n)\big\}
\end{displaymath}
\begin{equation}
~~~~~~\leq \sup_{0\leq m\leq n}\big\{\tau\cdot(n-m)-[a(n)-a(m)]\big\}
\label{eq:difficulty2}
\end{equation}
From Eq.(\ref{eq:difficulty2}), we have difficulty in further deriving more results if no additional constraint is added because we only know $P\{\tau\cdot(n-m)-[a(n)-a(m)]>x\}\leq h(x)$. When investigating the performance metrics such as delay bound and backlog bound in Section \ref{Sec-serviceguarantee}, we meet the similar difficulty.

\subsection{Virtual-system-delay Stochastic Arrival Curve}
The previous subsection stated the difficulty of applying i.a.t stochastic arrival curve to service guarantee analysis. This subsection introduces another stochastic arrival curve model that can help avoid such difficulty. This model is called \emph{virtual-system-delay} ($v.s.d$) stochastic arrival curve. The model explores the \emph{virtual system delay property} of deterministic arrival curve as implied by Lemma \ref{triplicity}.(2), which is that the amount of time a packet spends in a virtual SSQ fed with the same flow with a deterministic arrival curve is lower-bounded.  

For a flow having deterministic arrival curve, we construct a virtual SSQ system fed with the flow, which has infinite buffer space and the buffer is initially empty. Suppose the virtual SSQ system provides a deterministic service curve $\lambda$ to the flow or $d(n) = a\bar{\otimes}\lambda(n)$ for all $n\geq 0$. The amount of time packet $n$ spends in the virtual SSQ system is $W_s(n) = d(n) - a(n)$ = $\sup_{0\leq m\leq n}\{\lambda(n-m) - [a(n)-a(m)]\}$. If the flow is constrained by arrival curve $\lambda(n)-x$ for all $n\geq 0$, $W_s$ is also lower-bounded by $x$. 

Based on the virtual system time property, we define virtual-system-delay (v.s.d) stochastic arrival curve to characterize the arrival process.
\begin{definition}\label{vstarrcurve}
\textbf{(v.s.d Stochastic Arrival Curve).} A flow is said to have a virtual-system-delay (v.s.d) stochastic arrival curve $\lambda\in\mathcal{G}$ with bounding function $h\in\bar{\mathcal{G}}$, denoted by $a(n)\sim_{vd}\langle\lambda,h\rangle$, if for all $0\leq m\leq n$ and all $x\geq 0$, there holds
\begin{equation}
P\Big\{\sup_{0\leq m\leq n}\big\{\lambda(n-m) - [a(n) - a(m)]\big\} > x\Big\}\leq h(x).
\label{eq:vsdArrCurve}
\end{equation}
\end{definition}
Eq.(\ref{eq:vsdArrCurve}) can also be written as
\begin{equation}
~~~~~~~~~~P\big\{a\bar{\otimes}\lambda(n) - a(n) > x\big\}\leq h(x).
\label{eq:vsdArrCurve1}
\end{equation}
$a\bar{\otimes}\lambda(n)$ can be considered as the expected time that the packet would arrive to the system if the flow had passed through the virtual SSQ with service curve $\lambda(n)$. $x$ denotes the difference between the expected arrival time and the actual arrival time. Eq.(\ref{eq:vsdArrCurve1}) characterizes this difference $x$ by introducing a bounding function $h(x)$.

\textbf{Example 4.} Consider a flow with the same fixed packet size. Suppose all packet inter-arrival times are exponentially distributed with mean 
$\frac{1}{\mu}$. Based on the steady-state probability mass function (PMF) of the queue-waiting time for an M/D/1 queue \cite{Shortle:MD1queue}, we say that the flow has a v.s.d stochastic arrival curve $a(n)\sim_{vd}\langle D\cdot n,h^{exp}\rangle$ for $\forall D<\frac{1}{\mu}$, with $\rho=\mu\cdot D$ and  
\begin{displaymath}
h^{exp}(x) = 1 - (1-\rho)\sum_{i=0}^{\lfloor x/D\rfloor+1}e^{-\mu(-x)}\frac{[\mu(-x)]^i}{i!}
\end{displaymath}
where, $\lfloor x/D\rfloor$ denotes the greatest integer less than or equal to $x/D$. 

The following theorem establishes relationships between i.a.t stochastic arrival curve and v.s.d stochastic arrival curve. 
\begin{theorem}\label{it2vsdrelation}
\begin{enumerate}
\item If a flow has a v.s.d stochastic arrival curve $\lambda\in\mathcal{G}$ with bounding function $h\in\bar{\mathcal{G}}$, then the flow has an i.a.t stochastic arrival curve $\lambda\in\mathcal{G}$ with the same bounding function $h\in\bar{\mathcal{G}}$.
\item Inversely, if a flow has an i.a.t stochastic arrival curve $\lambda\in\mathcal{G}$ with bounding function $h\in\bar{\mathcal{F}}$, it also has a v.s.d stochastic arrival curve $\lambda_{-\eta}\in\mathcal{G}$ with bounding function $h^{\eta}\in\bar{\mathcal{G}}$ where
\begin{displaymath}
\lambda_{-\eta}(n) = \lambda(n) - \eta\cdot n~~~~~~~~~~~~
\end{displaymath} 
\begin{displaymath}
h^{\eta}(x) = \Big[h(x)+\frac{1}{\eta}\int_{x}^{\infty}h(y)dy\Big]_1
\end{displaymath}
for $\forall\eta > 0$. 
\end{enumerate}
\end{theorem}

\begin{proof}
The first part follows from that 
\begin{displaymath}
\lambda(n-m)-[a(n)-a(m)]\leq\sup_{0\leq m\leq n}\{\lambda(n-m)-[a(n)-a(m)]\}
\end{displaymath}
holds for $\forall 0\leq m \leq n$. 
For the second part, there holds
\begin{displaymath}
\sup_{0\leq m\leq n}\big\{\lambda_{-\eta}(n-m)-[a(n)-a(m)]\big\}~~~~~~~~
\end{displaymath}
\begin{displaymath}
~~\leq_{st} \sup_{0\leq m\leq n}\big\{\lambda_{-\eta}(n-m)-[a(n)-a(m)]\big\}^+
\end{displaymath}
Since for $\forall x\geq 0$, 
\begin{displaymath}
P\big\{\{\lambda(n-m)-\eta\cdot(n-m)-[a(n)-a(m)]\}^+>x\big\}
\end{displaymath}
\begin{displaymath}
=P\big\{\{\lambda(n-m)-\eta\cdot(n-m)-[a(n)-a(m)]\}>x\big\}
\end{displaymath}
\begin{displaymath}
\leq h\big(x+\eta\cdot(n-m)\big),
\end{displaymath}
 we have
\begin{displaymath}
P\Big\{\sup_{0\leq m\leq n}\{\lambda_{-\eta}(n-m)-[a(n)-a(m)]\}> x\Big\}
\end{displaymath}
\begin{displaymath}
\leq \sum_{m=0}^nP\Big\{\{\lambda_{-\eta}(n-m)-[a(n)-a(m)]\}^+> x\Big\}
\end{displaymath}
\begin{displaymath}
\leq \sum_{m=0}^nh(x+\eta\cdot(n-m))=\sum_{k=0}^nh(x+\eta\cdot k)~~~~~~~
\end{displaymath}
\begin{displaymath}
\leq \sum_{k=0}^{\infty}h(x+\eta\cdot k)=h(x)+\sum_{k=1}^{\infty}h(x+\eta\cdot k)~~~~~~~
\end{displaymath}
\begin{equation}
\leq h(x)+\frac{1}{\eta}\int_{x}^{\infty}h(y)dy.~~~~~~~~~~~~~~~~~~~~~~~~~~~
\label{eq:intermediatederivation}
\end{equation}
which is meaningful only when Eq.(\ref{eq:intermediatederivation}) is upper-bounded by one. The 1-fold integration of $h(x)$ is bounded by one because the condition assumes $h\in\bar{\mathcal{F}}$ as for the \cite{David:SBB}. Then the second part follows from Eq.(\ref{eq:intermediatederivation}). 
\end{proof}

Note that in the second part of the above theorem, $h(x)\in\bar{\mathcal{F}}$ while not $\in\bar{\mathcal{G}}$. If the requirement on the bounding function is relaxed to $h(x)\in\bar{\mathcal{G}}$, the above relationship may not hold in general.

The v.s.d stochastic arrival curve has a counterpart defined in the space-domain, the v.b.c stochastic arrival curve as defined in Definition \ref{vbcArrCurve}.~The following theorem establishes relationships between these two models.  
\begin{theorem}\label{vbctovsd}
\begin{enumerate}
\item If a flow has a v.b.c stochastic arrival curve $\alpha(t)\in\mathcal{G}$ with bounding function $f(x)\in\bar{\mathcal{G}}$, the flow has a v.s.d stochastic arrival curve $\lambda(n)\in\mathcal{G}$ with bounding function $h(y)\in\bar{\mathcal{G}}$, where $\lambda(n)=\inf\{\tau:\alpha(\tau)\geq n\}$ and $h(y) = f\big(\sup_{\tau\geq 0}[\alpha(\tau+y)-\alpha(\tau)+1]\big)$. 
\item If a flow has a v.s.d stochastic arrival curve $\lambda(n)\in\mathcal{G}$ with bounding function $h(y)\in\bar{\mathcal{G}}$, the flow has a v.b.c stochastic arrival curve $\alpha(t)\in\mathcal{G}$ with bounding function $f(x)\in\bar{\mathcal{G}}$, where $\alpha(t)=\sup\{k:\lambda(k)\leq t\}$ and $f(x) = h\big(\sup_{k\geq 0}[\lambda(k)-\lambda(k-x)]\big)$. 
\end{enumerate}
\end{theorem} 

\begin{proof}
(1) From Lemma \ref{spacedomain}, we know that for $\forall x,t\geq 0$, event $\{\mathcal{A}(t)\leq\mathcal{A}\otimes\alpha(t)+x\}$ implies event $\{a(n)\geq a\bar{\otimes}\lambda(n)-y\}$ where $y$ is obtained from Eq.(\ref{eq:yinversefunction}). Thus, there holds 
\begin{displaymath}
P\{\mathcal{A}(t)\leq\mathcal{A}\otimes\alpha(t)+x\}\leq P\{a(n)\geq a\bar{\otimes}\lambda(n)-y\}.
\end{displaymath}
We further have
\begin{displaymath}
P\{\mathcal{A}(t)>\mathcal{A}\otimes\alpha(t)+x\}\geq P\{a(n)<a\bar{\otimes}\lambda(n)-y\}.
\end{displaymath} 
From the condition that the flow has a v.b.c stochastic arrival curve $\alpha(t)$, we know $P\{\mathcal{A}(t)>\mathcal{A}\otimes\alpha(t)+x\}\leq f(x)$. According to Eq.(\ref{eq:xvalue}), we obtain
\begin{displaymath}
P\{a(n)<a\bar{\otimes}\lambda(n)-y\}\leq f\big(\sup_{\tau\geq 0}[\alpha(\tau+y)-\alpha(\tau)+1]\big).
\end{displaymath}

(2) From Lemma \ref{timedomainlemma}, we know that for $\forall y\geq 0$, event $\{a(n)\geq a\bar{\otimes}\lambda(n)-y\}$ implies event $\{\mathcal{A}(t)\leq\mathcal{A}\otimes\alpha(t)+x\}$ where $x$ is obtained from Eq.(\ref{eq:xvalue}). Thus, there holds
\begin{displaymath}
P\{a(n)\geq a\bar{\otimes}\lambda(n)-y\}\leq P\{\mathcal{A}(t)\leq\mathcal{A}\otimes\alpha(t)+x\}
\end{displaymath}
We further have 
\begin{displaymath}
P\{a(n)< a\bar{\otimes}\lambda(n)-y\}\geq P\{\mathcal{A}(t)>\mathcal{A}\otimes\alpha(t)+x\}
\end{displaymath} 
From the condition that the flow has a v.s.d stochastic arrival curve $\lambda(n)$, we know $P\{a(n)< a\bar{\otimes}\lambda(n)-y\}\leq h(y)$. According to Eq.(\ref{eq:yinversefunction}), we have
\begin{displaymath}
P\{\mathcal{A}(t)>\mathcal{A}\otimes\alpha(t)+x\}\leq h\big(\sup_{k\geq 0}[\lambda(k)-\lambda(k-x)]\big)
\end{displaymath} 
and complete the proof.
\end{proof}

\subsection{Maximum-(virtual)-system-delay Stochastic Arrival Curve}
The maximum-(virtual)-system-delay (m.s.d) stochastic arrival curve explores the \emph{maximum virtual system delay property} of deterministic arrival curve implied by Lemma \ref{triplicity}.(3), which is that the maximum system delay of a virtual SSQ fed with the same flow with a deterministic arrival curve is lower-bounded. 

Similar to the discussion for v.s.d stochastic arrival curve, for a flow having arrival curve, we construct a virtual SSQ system fed with the flow, which has infinite buffer space and the buffer is initially empty. Suppose the virtual SSQ system provides a deterministic service curve $\lambda$ to the flow or $d(n)=a\bar{\otimes}\lambda(n)$ for all $n\geq 0$. The maximum system delay in the virtual SSQ system for the first $n$ arrival packets as $\sup_{0\leq m\leq n}W_s(m)=\sup_{0\leq m\leq n}\sup_{0\leq q\leq m}\{\lambda(m-q)-[a(m)-a(q)]\}$. If the flow is constrained by arrival curve $\lambda(n)-x$ for all $n\geq 0$, the maximum system delay in the virtual SSQ is also upper-bounded by $x$.  

Based on the maximum virtual system delay property, we define m.s.d stochastic arrival curve model.
\begin{definition}\label{mvstarrcurve}
\textbf{(m.s.d Stochastic Arrival Curve).}

A flow is said to have a maximum-(virtual)-system-delay (m.s.d) stochastic arrival curve $\lambda(n)\in\mathcal{G}$ with bounding function $h(x)\in\bar{\mathcal{G}}$, denoted by $a(n)\sim_{md}\langle\lambda,h\rangle$, if for all $0\leq m\leq n$ and all $x\geq 0$, there holds
\begin{equation}
P\Big\{\sup_{0\leq m\leq n}\sup_{0\leq q\leq m}\big\{\lambda(m-q) - [a(m) - a(q)]\big\} > x\Big\}\leq h(x).
\end{equation}
\end{definition}

\subsection{Deterministic Service Curve}
To provide service guarantees to an arrival-constrained flow $F$,~the system usually needs to allocate a minimum service rate to $F$.~A guaranteed minimum service rate is equivalent to a guaranteed maximum service time for each packet of the flow, and accordingly the packet's departure time from the system is bounded.~Because packets of the same flow are served in FIFO manner, any packet $p^n$ from this flow will depart by $\hat{d}(n)$ which is iteratively defined by  
\begin{equation}
\hat{d}(n) = \max[a(n),\hat{d}(n-1)] + \delta(n)
\label{eq:departuretime}
\end{equation}
with $\hat{d}(0)=0$, where $\delta(n)$ is the service time guaranteed to $p^n$. By applying Eq.(\ref{eq:departuretime}) iteratively to its right-hand side, it becomes
\begin{equation}
~~~~\hat{d}(n) = \sup_{0\leq m \leq n}[a(m) + \sum_{i=m}^n\delta(i)] 
\label{eq:departuretimeori}
\end{equation}
where $\sum_{i=m}^n\delta(i)$ is the guaranteed cumulative service time for packet $p^m$ to $p^n$. Suppose we can use a function $\gamma(n-m)$ to denote $\sum_{i=m}^n\delta(i)$, i.e. $\gamma(n-m)= \sum_{i=m}^n\delta(i)$. Then, Equation(\ref{eq:departuretimeori}) becomes
$$
\hat{d}(n) = \sup_{0\leq m \leq n}[a(m) + \gamma(n-m)] = a\bar{\otimes}\gamma(n)
$$
which provides a basis for the following {\em time-domain} (deterministic) server model that charaterizes the service using an upper bound on the cumulative packet service time \cite{Chang:MaxPlus}:

\begin{definition}\label{DetServerModel}
\textbf{(Service Curve).} Consider a system $\mathcal{S}$ with input process $a(n)$ and output process $d(n)$. The system is said to provide to the input a (deterministic) service curve $\gamma(n)\in\mathcal{G}$, if for $\forall n\geq 0$, 
\begin{equation}
~~~~~~~d(n)\leq a\bar{\otimes}\gamma(n).
\label{eq:defnideptime}
\end{equation}
\end{definition}

The (deterministic) service curve model has the following duality principle: 
\begin{lemma}\label{duality}
For $\forall x\geq 0$, $d(n)-a\bar{\otimes}\gamma(n)\leq x$ for all $n\geq 0$, if and only if $\sup_{0\leq m\leq n}[d(n)-a\bar{\otimes}\gamma(n)]\leq x$ for $\forall n\geq 0$, where $\gamma\in\mathcal{G}$.
\end{lemma}
\begin{proof}
For the "if" part, it holds because $d(n)-a\bar{\otimes}\gamma(n)\leq\sup_{0\leq m\leq n}[d(n)-a\bar{\otimes}\gamma(n)]$. For the "only if" part, from $d(n)-a\bar{\otimes}\gamma(n)\leq x$ for $\forall n\geq 0$, we have $\sup_{0\leq m\leq n}[d(n)-a\bar{\otimes}\gamma(n)]\leq\sup_{0\leq m\leq n}[x]=x$.
\end{proof}
By the definition of service curve, the first part of Lemma \ref{duality} defines a service curve $\gamma(n)+x$. Lemma \ref{duality} states that if a server provides service curve $\gamma(n)+x$, then $\sup_{0\leq m\leq n}[d(m)-a\bar{\otimes}\gamma(m)]\leq x$ holds, and vice versa.~In this sense, we call Lemma \ref{duality} the \emph{duality principle} of service curve. 

\subsection{Stochastic Service Curve}
For networks providing stochastic service guarantees, following the principle of Eq.(\ref{eq:departuretimeori}), we have the following expression for the expected departure time of packet $p^n$
\begin{displaymath}
\hat{d}(n) = \sup_{0\leq m\leq n}[a(m)+\sum_{i=m}^n(\delta(i)+\epsilon(i))]
\end{displaymath}
where we assume $\delta(i)$ is the deterministic part while $\epsilon(i)$ the random part in the total service time $\delta(i)+\epsilon(i)$ guaranteed to packet $p^{i}$. We call $\epsilon(i)$ \emph{stochastic error term} associated to $\delta(i)$. Here, $\epsilon(n)$ is introduced to represent the additional delay of $p^n$ due to some randomness. For example, an error-prone wireless link is often considered to operate in two states. If the link is in \lq good\rq~condition, it can send and receive data correctly; if the link is in \lq bad\rq~condition due to errors, the data that should be sent immediately has to be queued longer until the channel changes to \lq good\rq~condition. Then, $\epsilon(n)$ in this case represents the time period in which the channel is in \lq bad\rq~condition between the time when $p^{n-1}$ has been sent correctly and the time when $p^n$ can be sent.

With the consideration of the stochastic error term, the (deterministic) service curve can be extended to a stochastic version as follows: 
\begin{definition}\label{StoSerCurve}
\textbf{(i.d Stochastic Service Curve).} 

A system is said to provide an \emph{inter-departure time (i.d) stochastic service curve} $\gamma\in\mathcal{G}$ with bounding function $j\in\bar{\mathcal{G}}$, denoted by $\mathcal{S}\sim_{id}\langle\gamma,j\rangle$, if for all $n\geq 0$ and all $x\geq 0$, there holds
\begin{equation}
P\Big\{d(n) - a\bar{\otimes}\gamma(n) > x\Big\}\leq j(x).
\label{eq:idservicecurve}
\end{equation} 
\end{definition}

\textbf{Example 5.} Consider two nodes, the sender and the receiver, communicate through an error-prone wireless link. Packets have fixed-length.~Packets arriving to the sender node are served in FIFO manner.~Assume the guaranteed per-packet service time is $\delta$ without any error.~To simplify the analysis, assume the time slot length equals $\delta$.~The sender sends packets correctly only when the link is in \lq good\rq~ condition.~If the link is in \lq bad\rq~condition, no packets can be sent correctly. In addition, the sender can send the head-of-queue packet only at the beginning of a time slot, i.e., the time period during which the link is in \lq bad\rq condition should be an integer times of $\delta$. The probability that a packet can be sent correctly is determined by packet error rate (PER). PER is determined by the packet length and the bit error rate (BER). Here, we assume packet errors happen independently and the same PER denoted by $P_e$ is applied to all packets. The successful transmission probability of one packet is hence $1-P_e$.

Suppose $P\{\Delta(n)=i\}=P_e^{i-1}(1-P_e)$, $i\geq 1$, where $\Delta(n)$ represents the number of time slots necessary to successfully send the $n^{th}$ packet with respect to the successful transmission probability $1-P_e$. The number of time slots necessary to successfully send $n$ packets is $\sum_{k=1}^n\Delta(k)$ which has the negative binomial distribution
\begin{displaymath}
P\Big\{\sum_{k=1}^n\Delta(k)=i\Big\}=\left\{\begin{array}{ll}\begin{pmatrix}i-1\\n-1\end{pmatrix}(1-P_e)^nP_e^{i-n}, ~~i\geq n\\
0, ~~~~~~~~~~~~~~~~~~~~~~~~~~~~~~i < n\end{array}\right.
\end{displaymath}  
Then the sender provides to its input a stochastic service curve $\gamma$ which has the following distribution
\begin{displaymath}
P\{\gamma(n)=\lceil\frac{\tau}{\delta}\rceil\}=\begin{pmatrix}\lceil\frac{\tau}{\delta}\rceil-1\\n-1\end{pmatrix}(1-P_e)^nP_e^{\lceil\frac{\tau}{\delta}\rceil-n}
\end{displaymath} 
where $\tau$ is the guaranteed service time to successfully send $n$ packets and $\lceil x\rceil$ denotes the smallest integer greater than or equal to $x$. 

We can find $n_0\leq n$ such that $a\bar{\otimes}\gamma(n)$ takes its maximum value, i.e., $a\bar{\otimes}\gamma(n)=a(n_0)+\gamma(n-n_0+1)$. From Eq.(\ref{eq:idservicecurve}), we have
\begin{displaymath}
P\{\gamma(n-n_0+1) < d(n)-a(n_0)- x\}\leq j(x)
\end{displaymath}
where 
\begin{displaymath}
j(x)=\sum_{i=n}^{\lceil\frac{d(n)-a(n_0)-x}{\delta}\rceil-1}\begin{pmatrix}i-1\\n-1\end{pmatrix}(1-P_e)^nP_e^{i-n}
\end{displaymath}

In Sec.~\ref{Sec-Property}, we show that many results can be derived from the i.d stochastic service curve model. However, without additional constraints, we have difficulty in proving the concatenation property for i.d stochastic service curve. To address this difficulty, we introduce a stronger definition in the following subsection.

\subsection{Constrained Stochastic Service Curve}
The constrained stochastic service curve model is generalized from the (deterministic) service curve model based on its duality principle. From Lemma \ref{duality}, we know that a system with input $a(n)$ and output $d(n)$ has a service curve $\gamma(n)$ if and only if for $\forall n\geq 0$, 
\begin{equation}
~~~~\sup_{0\leq m\leq n}\{d(m) - a\bar{\otimes}\gamma(m)\}\leq x.
\label{eq:servicecurve}
\end{equation} 
Inequality (\ref{eq:servicecurve}) provides the basis to generalize the (deterministic) service curve model to the constrained stochastic service curve defined as follows:

\begin{definition}\label{ConstrainedSerCurve}
\textbf{(Constrained Stochastic Service Curve).} A system is said to provide a \emph{constrained stochastic service curve} (c.s) $\gamma\in\mathcal{G}$ with bounding function $j\in\bar{\mathcal{G}}$, denoted by $\mathcal{S}\sim_{cs}\langle\gamma,j\rangle$, if for $\forall n,x\geq 0$, there holds
\begin{equation}
P\Big\{\sup_{0\leq m\leq n}[d(m) - a\bar{\otimes}\gamma(m)] > x\Big\}\leq j(x).
\end{equation} 
\end{definition}

The following theorem establishes relationships between i.d stochastic service curve and c.s stochastic service curve. 
\begin{theorem}\label{stosercur2constosercur}
\begin{enumerate}
\item If a server $\mathcal{S}$ provides to its input $a(n)$ a c.s stochastic service curve $\gamma(n)$ with bounding function $j(x)\in\bar{\mathcal{G}}$, it provides to the input $a(n)$ an i.d stochastic service curve $\gamma(n)$ with the same bounding function $j(x)\in\bar{\mathcal{G}}$, i.e., $\mathcal{S}\sim_{id}\langle\gamma,j\rangle$;

\item If a server $\mathcal{S}$ provides to its input $a(n)$ an i.d stochastic service curve $\gamma(n)$ with bounding function $j(x)\in\bar{\mathcal{F}}$, it provides to the input $a(n)$ a c.s stochastic service curve $\gamma_{+\eta}(n)=\gamma(n)+\eta\cdot n$ with bounding function $j^{\eta}(x)\in\bar{\mathcal{F}}$ for $\forall\eta>0$, where 
\begin{displaymath}
j_{\eta}(x) = \Big[\frac{1}{\eta}\int_{x-\eta\cdot n}^n j(y)dy\Big]_1.
\end{displaymath}
\end{enumerate}
\end{theorem}
\begin{proof}
The first part follows since there always holds $d(n)-a\bar{\otimes}\gamma(n)\leq \sup_{0\leq m\leq n}\{d(m)-a\bar{\otimes}\gamma(m)\}$. 

For the second part, there holds for $\forall 0\leq m\leq n$, 
\begin{displaymath}
a\bar{\otimes}\gamma_{+\eta}(m) \geq a\bar{\otimes}\gamma(m) + \eta\cdot m - \eta\cdot n 
\end{displaymath}
and then 
\begin{displaymath}
d(m) - a\bar{\otimes}\gamma_{+\eta}(m) \leq d(m) - a\bar{\otimes}\gamma(m) - \eta\cdot m + \eta\cdot n. 
\end{displaymath}
Thus, we obtain
\begin{displaymath}
P\big\{\sup_{0\leq m\leq n}\{d(m)-a\bar{\otimes}\gamma_{+\eta}(m)\} > x\big\}~~~~~~~~~~~~~~~~~~~~~~~
\end{displaymath}
\begin{displaymath}
\leq P\big\{\sup_{1\leq m\leq n}\big[d(m)-a\bar{\otimes}\gamma(m) - \eta(m)\big]^{+} > x - \eta\cdot n\big\}
\end{displaymath}
for which when $x - \eta\cdot n<0$, the right hand side is equal to 1. In the following, we assume $x - \eta\cdot n\leq 0$ under which, there holds
\begin{displaymath}
P\big\{\sup_{0\leq m\leq n}\{d(m)-a\bar{\otimes}\gamma_{+\eta}(m)\} > x\big\}~~~~~~~~~~~~~~~~
\end{displaymath}
\begin{displaymath}
\leq \sum_{m=1}^n P\big\{\big[d(m)-a\bar{\otimes}\gamma(m) - \eta(m)\big] > x - \eta\cdot n\big\}
\end{displaymath}
\begin{displaymath}
\leq \sum_{m=1}^nj(x - \eta\cdot n + \eta\cdot m) \leq \frac{1}{\eta}\int_{x-\eta\cdot n}^n j(y)dy.~~~~
\end{displaymath}
Because the probability is always not greater than 1, the second part follows from the above inequality. 
\end{proof}
In the second part of the above theorem, $j(x)\in\bar{\mathcal{F}}$ while not $\in\bar{\mathcal{G}}$. If the requirement on the bounding function is relaxed to $j(x)\in\bar{\mathcal{G}}$, the above relationship may not hold in general.

\section{Basic Properties}\label{Sec-Property}
This section presents results derived from the time-domain traffic models and server models introduced in Sec.~\ref{Section-Model}. Particularly, we investigate the five basic properties introduced in Sec.~\ref{Sec-intro}, which are service guarantees including delay bound and backlog bound, output characterization, concatenation property and superposition property. However, some properties can directly be proved only for the combination of a specific traffic model and a specific server model.~This explains why we need to establish the various relationships between models in Sec.~\ref{Section-Model}. With these relationships, we can obtain the corresponding results for models which we are interested in.  

\subsection{Service Guarantees}\label{Sec-serviceguarantee}
This subsection investigates probabilistic bounds on delay and backlog under the combination of v.s.d stochastic arrival curve and i.d stochastic service curve. 

We start with deriving the bound on delay that a packet would experience in a system.  
\begin{theorem}\label{Delaybound}
\textbf{(Delay Bound).}
Consider a system $\mathcal{S}$ providing an i.d stochastic service curve $\gamma\in\mathcal{G}$ with bounding function $j\in\bar{\mathcal{G}}$ to the input which has a v.s.d arrival curve $\lambda\in\mathcal{G}$ with bounding function $h\in\bar{\mathcal{G}}$. Let $D(n) = d(n) - a(n)$ be the delay in the system of the $n^{th}(\geq0)$ packet. For $\forall x\geq 0$, $D(n)$ is bounded by
\begin{equation}
P\{D(n) > x\} \leq j\otimes h(x - \gamma\oslash\lambda(0)).
\label{eq:DelayBound}
\end{equation} 
\end{theorem}
\begin{proof}
For $\forall n\geq 0$, there holds 
\begin{displaymath}
d(n) - a(n) = \big[d(n) - a\bar{\otimes}\gamma(n)\big] + \big[a\bar{\otimes}\gamma(n) - a(n)\big]~~~~~~~~~~
\end{displaymath}
\begin{displaymath}
= \big[d(n) - a\bar{\otimes}\gamma(n)\big] + \sup_{0 \leq m \leq n}\big\{\lambda(n-m) - \big[a(n) - a(m)\big] 
\end{displaymath}
\begin{displaymath}
~~+ \gamma(n-m) - \lambda(n-m)\big\}
\end{displaymath}
\begin{displaymath}
\leq \big[d(n) - a\bar{\otimes}\gamma(n)\big] + \sup_{0\leq m\leq n}\big\{\lambda(n - m) - \big[a(n) - a(m)\big]\big\} 
\end{displaymath}
\begin{displaymath}
~+ \sup_{0\leq m\leq n}\big\{\gamma(n-m)-\lambda(n-m)\big\}
\end{displaymath}
\begin{displaymath}
\leq \big[d(n) - a\bar{\otimes}\gamma(n)\big] + \sup_{0\leq m\leq n}\big\{\lambda(n - m) - \big[a(n) - a(m)\big]\big\} 
\end{displaymath}
\begin{equation}
~ + \sup_{k\geq 0}\big\{\gamma(k) - \lambda(k)\big\}.
\label{eq:proof1}
\end{equation}
The right-hand side of Eq.(\ref{eq:proof1}) implies a sufficient condition to obtain $P\{D(n) > x\}$, which is that $P\big\{d(n) - a\bar{\otimes}\gamma(n) > x\big\}$ and $P\Big\{\sup_{0\leq m\leq n}\big\{\lambda(n - m) - \big[a(n) - a(m)\big]\big\} > x\Big\}$ are known. To ensure the system's stability, we should also have
\begin{equation}
~~~~~~~~~~\lim_{k\to \infty}\frac{1}{k}[\gamma(k) - \lambda(k)] \leq 0.
\label{eq:systemstable}
\end{equation} 
In the rest of the paper, without explicitly stating, we shall assume inequality (\ref{eq:systemstable}) holds.
From Lemma \ref{lemma1} and $\sup_{k\geq 0}\big\{\gamma(k)\\-\lambda(k)\big\}=\gamma\oslash\lambda(0)$, we conclude 
\begin{displaymath}
P\{D(n) > x\} \leq j\otimes h(x-\gamma\oslash\lambda(0)).~~~~~~~~~~~~~~~~
\end{displaymath}
\end{proof}

Next, we consider backlog bound of a system. By definition, the backlog in the system at time $t\geq 0$ is $\mathcal{B}(t)=\mathcal{A}(t)-\mathcal{A}^*(t)$. If $a(n)$ is the arrival time of the latest packet arriving to the system by time $t$, then $\mathcal{B}(t)$ is
\begin{equation}
\mathcal{B}(t) \leq \inf\big\{k\geq 0: d(n-k) \leq a(n)\big\}.
\label{eq:backlogDef}
\end{equation}
Eq.(\ref{eq:backlogDef}) implies that, for $\forall x \geq 0$, if $\mathcal{B}(t)> x$, there must be $a(n)<d(n-x)$. Thus event $\{\mathcal{B}(t) > x\}$ implies event $\{a(n) < d(n-x)\}$ and $P\{\mathcal{B}(t) > x\} \leq P\{a(n) < d(n-x)\}$. Then we have the following result for backlog. 
\begin{theorem}\label{BacklogBound}
\textbf{(Backlog Bound).}
Consider a system $\mathcal{S}$ providing an i.d stochastic service curve $\gamma\in\mathcal{G}$ with bounding function $j\in\bar{\mathcal{G}}$ to the input which has a v.s.d stochastic arrival curve $\lambda\in\mathcal{G}$ with bounding function $h\in\bar{\mathcal{G}}$. The backlog at time $t$ ($t\geq 0$), $\mathcal{B}(t)$, is bounded by
\begin{equation}
P\{\mathcal{B}(t) > H(\lambda,\gamma+x)\} \leq j\otimes h(x)
\end{equation}
for any $x\geq 0$, where, $H(\lambda,\gamma+x)=\sup_{n\geq 0}\big\{\inf[k\geq 0:\gamma(n-k)+x\leq\lambda(n)]\big\}$ is the maximum horizontal distance between functions $\lambda(n)$ and $\gamma(n)+x$ for $\forall x\geq 0$.
\end{theorem}
\begin{proof}
Similar to prove the delay bound, we have
\begin{displaymath}
d(n-x) - a(n) = [d(n-x) - a\bar{\otimes}\gamma(n-x)] + [a\bar{\otimes}\gamma(n-x) - a(n)]
\end{displaymath} 
\begin{displaymath}
= \big[d(n-x) - a\bar{\otimes}\gamma(n-x)\big] + \sup_{0\leq k\leq n-x}\big\{\lambda(n-k) - [a(n) - a(k)]  
\end{displaymath}
\begin{displaymath}
+~ \gamma(n-x-k)-\lambda(n-k)\big\}~~~~~~~~~~
\end{displaymath}
\begin{displaymath}
\leq \big[d(n-x) - a\bar{\otimes}\gamma(n-x)\big] + \sup_{0\leq k\leq n-x}\big\{\lambda(n-k) - \big[a(n+x)~~
\end{displaymath}
\begin{displaymath}
 - a(k)\big]\big\} + \sup_{0\leq k\leq n-x}\big\{\gamma(n-x-k) - \lambda(n-k)\big\}~~~~~~~~
\end{displaymath}
Let $v = n-k$. The above inequality is written as
\begin{displaymath}
d(n-x) - a(n) \leq \big[d(n-x) - a\bar{\otimes}\gamma(n-x)\big] + \sup_{0\leq k\leq n}\big\{\lambda(n-k) 
\end{displaymath}
\begin{displaymath}
~~~~~~~~~~~~~~~~~~~- \big[a(n) - a(k)\big]\big\} + \sup_{x\leq v\leq n}\big\{\gamma(v-x) - \lambda(v)\big\}
\end{displaymath}
Let $x = H(\lambda,\gamma+y)$, we have
\begin{displaymath}
d\big(n-h(\lambda,\gamma+y)\big) - a(n)\leq \big[d\big(n-h(\lambda,\gamma+y)\big)-~~~~~~~~~~~~~~~~ 
\end{displaymath}
\begin{equation}
a\bar{\otimes}\gamma\big(n-h(\lambda,\gamma+y)\big)\big] + \sup_{0\leq k\leq n}\big\{\lambda(n-k) - [a(n) - a(k)]\big\} -y
\label{eq:derBacklog}
\end{equation}
Under the same conditions as analyzing the delay, we obtain
\begin{displaymath}
P\{\mathcal{B}(t) > H(\lambda,\gamma+x)\} \leq j\otimes h(x).
\end{displaymath}
\end{proof}

\subsection{Output Characterization}
This subsection presents the result for characterizing the departure process from a system. 
\begin{theorem}\label{Outputchar1}
\textbf{(Output Characterization).} Consider a system $\mathcal{S}$ provides an i.d stochastic service curve $\gamma(n)\in\mathcal{G}$ with bounding function $j(x)\in\bar{\mathcal{G}}$ to its input which has a v.s.d stochastic arrival curve $\lambda(n)\in\mathcal{G}$ with bounding function $h(x)\in\bar{\mathcal{G}}$. The output has an i.a.t stochastic arrival curve $\lambda\bar{\oslash}\gamma(n-m)$ with bounding function $j\otimes h(x)\in\bar{\mathcal{G}}$.
\end{theorem}
\begin{proof}
For any two departure packets $m<n$, there holds 
\begin{displaymath}
d(n) - d(m) \geq a(n) - a\bar{\otimes}\gamma(m) + a\bar{\otimes}\gamma(m) - d(m) ~~~~~~~~~~~~~~~~
\end{displaymath}
\begin{displaymath}
-\big[d(n) - d(m)\big] \leq \big[d(m) - a\bar{\otimes}\gamma(m)\big] + a\bar{\otimes}\gamma(m) - a(n)~~~~~~~~~~~~
\end{displaymath}
\begin{displaymath}
\leq \big[d(m) - a\bar{\otimes}\gamma(m)\big] + \sup_{0\leq k\leq n}\big\{\lambda(n-k) - [a(n) - a(k)]\big\}~~~~~~ 
\end{displaymath}
\begin{displaymath}
+ \sup_{0\leq v \leq m}\big\{\gamma(v) - \lambda(n-m+v)\big\}~~~~~~~~~~~~~~~~~~~~~~~~~~~~~~~
\end{displaymath}
\begin{displaymath}
= \big[d(m) - a\bar{\otimes}\gamma(m)\big] + \sup_{0\leq k\leq n}\big\{\lambda(n-k) - [a(n) - a(k)]\big\}~~~~~~~~~ 
\end{displaymath}
\begin{displaymath}
- \inf_{0\leq v \leq m}\big\{\lambda(n - m + v) - \gamma(v)\big\}~~~~~~~~~~~~~~~~~~~~~~~~~~~~~
\end{displaymath}
Adding $\inf_{0\leq v \leq m}\big\{\lambda(n - m + v) - \gamma(v)\big\}$ to both sides of the above inequality, we get
\begin{displaymath}
\inf_{0 \leq v \leq m}\big\{\lambda(n - m + v) - \gamma(v)\big\} - [d(n) - d(m)]~~~~~~~~~~~~~~~~~~~
\end{displaymath}
\begin{displaymath}
\leq \big[d(m) - a\bar{\otimes}\gamma(m)\big] + \sup_{0\leq k\leq n}\big\{\lambda(n-k) - [a(n) - a(k)]\big\}    
\end{displaymath}
With the same conditions as analyzing delay, we conclude
\begin{displaymath}
P\Big\{\lambda\bar{\oslash}\gamma(n-m) - [d(n) - d(m)] > x\Big\} \leq j\otimes h(x)
\end{displaymath}
\end{proof}

\subsection{Concatenation Property}\label{Concatenation}
The concatenation property uses an equivalent system to represent a system of multiple servers connected in tandem, each of which provides stochastic service curve to the input. Then the equivalent system provides the input a stochastic service curve, which is derived from the stochastic service curve provided by all involved individual servers. 

\begin{theorem}\label{concatenation}
\textbf{(Concatenation Property).} Consider a flow passing through a network of $N$ systems in tandem. If each system $k(= 1,2,...,N)$ provides a c.s stochastic service curve $\mathcal{S}^k\sim_{cs}\langle\gamma^k,j^k\rangle$ to its input, then the network guarantees to the flow a c.s stochastic service curve $\mathcal{S}\sim_{cs}\langle\gamma,j\rangle$ with 
\begin{equation}
\gamma(n) = \gamma^1\bar{\otimes}\gamma^2\bar{\otimes}\cdot\cdot\cdot\bar{\otimes}\gamma^N(n)
\end{equation} 
\begin{equation}
j(x) = j^1\otimes j^2\otimes\cdot\cdot\cdot\otimes j^N(x).
\end{equation}
\end{theorem}
\begin{proof}
We shall only prove the two-node case, from which, the proof can be easily extended to the $N$-node case. The departure of the first node is the arrival to the second node, so $d^1(n) = a^2(n)$. In addition, the arrival to the network is the arrival to the first node, i.e., $a(n) = a^1(n)$, and the departure from the network is the departure from the second node, i.e., $d(n) = d^2(n)$, where, $a(n)$ and $d(n)$ denote the arrival process to and departure process from the network, respectively. We then have, 
\begin{displaymath}
\sup_{0\leq m \leq n}\{d(m) - a\bar{\otimes}\gamma^1\bar{\otimes}\gamma^2(m)\}
\end{displaymath}
\begin{equation}
= \sup_{0\leq m \leq n}\{d^2(m) - (a^1\bar{\otimes}\gamma^1)\bar{\otimes}\gamma^2(m)\} 
\label{eq:superdepart}
\end{equation}
Now let us consider any $m$, ($0\leq m\leq n$), for which we get,
\begin{displaymath}
d^2(m) - (a^1\bar{\otimes}\gamma^1)\bar{\otimes}\gamma^2(m)~~~~~~~~~~~~~~~~~~~~~~~~~~~~~~~~~~~~ ~~
\end{displaymath}
\begin{displaymath}
= d^2(m) - \sup_{0\leq k\leq m}\big\{a^1\bar{\otimes}\gamma^1(k) + \gamma^2(m-k) - d^1(k) + a^2(k)\big\} 
\end{displaymath}
\begin{displaymath}
= d^2(m) + \inf_{0\leq k\leq m}\big\{d^1(k) - a^1\bar{\otimes}\gamma^1(k) - \gamma^2(m-k) - a^2(k)\big\}    
\end{displaymath}
\begin{displaymath}
\leq \sup_{0\leq k\leq m}\{d^1(k)-a^1\bar{\otimes}\gamma^1(k)\}+d^2(m)~~~~~~~~~~~~~~~~~~~~~~~~~~~~~
\end{displaymath}
\begin{displaymath}
+\inf_{0\leq k\leq m}\big\{- [a^2(k) + \gamma^2(m-k)]\big\} ~~~~~~~~~~~~~~~~
\end{displaymath}
\begin{equation}
\leq \sup_{0\leq k\leq m}\{d^1(k) - a^1\bar{\otimes}\gamma^1(k)\} + [d^2(m) - a^2\bar{\otimes}\gamma^2(m)]
\label{eq:Interstep}
\end{equation}
Applying Eq.(\ref{eq:superdepart}) to Eq.(\ref{eq:Interstep}), we obtain
\begin{displaymath}
\sup_{0\leq m \leq n}\{d^2(m) - (a^1\bar{\otimes}\gamma^1)\bar{\otimes}\gamma^2(m)\}~~~~~~~~~~~~~~~~~~~~~~ ~~~~~~~~~~
\end{displaymath}
\begin{equation}
\leq \sup_{0\leq k\leq n}\{d^1(k) - a^1\bar{\otimes}\gamma^1(k)\} + \sup_{0\leq m\leq n}\{d^2(m) - a^2\bar{\otimes}\gamma^2(m)\}
\end{equation}
with which, since both nodes provide c.s stochastic service curve to their input, the theorem follows from Lemma \ref{lemma1} and the definition of c.s stochastic service curve. 
\end{proof}

\subsection{Superposition Property}\label{Superposition}
The superposition property means that the superposition of flows can be represented using the same traffic model. With this property, the aggregate of multiple individual flows may be viewed as a single aggregate flow. Then the service guarantees for the aggregate flow can be derived in the same way as for a single flow. 

First, we only consider the aggregate of two flows, $F_1$ and $F_2$. Let $a_1(n)$, $a_2(n)$ and $a(n)$ be the arrival process of $F_1$, $F_2$ and the aggregate flow $F_A$, respectively.

For any packet $p^n$ of the aggregate flow $F_A$, it is either the $m^{th}$ packet from flow $F_1$ or the $(n-m)^{th}$
 packet from flow $F_2$, where $m\in[0,n]$, i.e.
\begin{displaymath}
a(n) = \max\{a_1(m),a_2(n-m)\}
\end{displaymath} 
For example, $a(1)$ is either $\max[a_1(0),a_2(1)]$ or $\max[a_1(1),a_2(0)]$ and the minimum of these two possibilities, 
\begin{displaymath}
a(1) = \inf\big\{\max[a_1(0),a_2(1)],\max[a_1(1),a_2(0)]\big\}.
\end{displaymath} 
We can see another example 
\begin{displaymath}
a(2) = \inf\big\{\max[a_1(0),a_2(2)],\max[a_1(1),a_2(1)],
\end{displaymath}
\begin{displaymath}
~~~~~~~~~~~~~~\max[a_1(2),a_2(0)]\big\}.
\end{displaymath}
Essentially, we have for any packet $n$ of the aggregate flow 
\begin{equation}
a(n) = \inf_{0\leq m \leq n}\big\{\max[a_1(m),a_2(n-m)]\big\}.
\label{eq:aggregatepkt}
\end{equation}
We generalize the result to the superposition of $N(\geq 2)$ flows
\begin{displaymath}
a(n) = \inf_{\sum m_i=n}\big\{\max[a_1(m_1),a_2(m_2),...,
\end{displaymath}
\begin{equation}
~~~~~~~~~~~~~~~~~~~~a_N(n-\sum_{i=1}^{N-1}m_i)]\big\}.
\label{eq:Naggregate}
\end{equation}

From Eq.(\ref{eq:Naggregate}), it is difficult to directly characterize the packet inter-arrival time of the aggregate flow. We know that if a flow has a v.s.d stochastic arrival curve, with Theorem \ref{vbctovsd}(2), this flow has a v.b.c stochastic arrival curve, for which the superposition property holds \cite{Jiang:book}. Thus, we can indirectly prove that the superposition property holds for the v.s.d stochastic arrival curve. 

If flow $i$ has a v.s.d stochastic arrival curve $a_i(n)\sim_{vd}\langle\lambda_i,h_i\rangle$ $i=1,2,...,N$, from Theorem \ref{vbctovsd}(2), flow $i$ has a v.b.c stochastic arrival curve $\alpha_i(t)$ with bounding function $f_i(x)=h_i\big(\sup_{\tau\geq 0}[\alpha_i(\tau+y)-\alpha_i(\tau)+1]\big)$ , where $\alpha_i(t)=\sup\{k:\lambda_i(k)\leq t\}$. According to Lemma \ref{superpositionMinPlus}, the aggregate flow has a v.b.c stochastic arrival curve $\alpha(t)=\sum_{i=1}^N\alpha_i(t)$ with bounding function $f(x)=f_1\otimes\cdot\cdot\cdot\otimes f_N(x)$. We apply Theorem \ref{vbctovsd}(1) and obtain the following result:
\begin{theorem}\label{superpositionmaxplus}
Consider $N$ flows with arrival processes 

$a_i(n)\sim_{vd}\langle\lambda_i,h_i\rangle$, $i=1,...,N$. For the aggregate of these flows, there holds $a(n)\sim_{vd}\langle\lambda,h\rangle$ with $\lambda(n)=\inf\{\tau:\alpha(\tau)\geq n\}$ and 
\begin{displaymath}
h(y)=f\Big(\sup_{k\geq 0}[\lambda(k)-\lambda(k-x)]\Big),
\end{displaymath}
where $\alpha(t) = \sum_{i=1}^N\alpha_i(t)$ and $f(x)=f_1\otimes\cdot\cdot\cdot\otimes f_N(x)$ with 
\begin{displaymath}
\alpha_i(t) = \sup\{k:\lambda_i(k)\leq t\}~~~~~~~~~~~~~~~~~~~~~~~~~~~
\end{displaymath}
\begin{displaymath}
f_i(x)=h_i\big(\sup_{\tau\geq 0}[\alpha_i(\tau+y)-\alpha_i(\tau)+1]\big).
\end{displaymath}
\end{theorem}

\subsection{Leftover Service Characterization}
This subsection explores the leftover service characterization under aggregate scheduling. To ease the discussion, we consider the simplest case when there are two flows competing resource in a system under FIFO aggregation. Suppose that if packets arrive to the system simultaneously, they are inserted into the FIFO queue randomly. Consider a system fed with a flow $F_A$ which is the aggregation of two constituent flows $F_1$ and $F_2$. Suppose both the service characterization from the server and traffic characterization from $F_2$ are known. We are interested in characterizing the service time received by $F_1$, with which per-flow bounds for $F_1$ can be then easily obtained using earlier results derived in the previous subsections. 

\begin{theorem}\label{leftoverservice}
Consider a system $\mathcal{S}$ with input $F_A$ that is the aggregation of two constituent flows $F_1$ and $F_2$. Suppose $F_2$ has a (deterministic) arrival curve $\lambda_2(n)\in\mathcal{G}$, and the system provides to the input an i.d stochastic service curve $\gamma\in\mathcal{G}$ with bounding function $j(x)\in\bar{\mathcal{G}}$. Then if ~$\gamma\big(n+\sup[q:\lambda_2(q)\leq a_1(n)]\big)\in\mathcal{G}$, $F_1$ receives an i.d stochastic service curve $\gamma\big(n+\sup[q:\lambda_2(q)\leq a_1(n)]\big)$ with the same bounding function $j(x)$. 
\end{theorem}
\begin{proof}
Suppose packet $p_1^n$ is the $(n+m)^{th}$ packet of $F_A$, i.e., $a(n+m)=a_1(n)$, where $m$ represents the number of packets from $F_2$. As the system provides an i.d stochastic service curve $\gamma(n)$ to the aggregate flow $F_A$, there holds
\begin{displaymath}
P\{d(n+m) - a\bar{\otimes}\gamma(n+m) > x\}\leq j(x). 
\end{displaymath}
$a_1(n)=a(n+m)$ indicates $a_2(m)\leq a_1(n)$. Let $\bar{m}=\sup[q:\lambda_2(q)\leq a_1(n)]$. As $\lambda_2$ is the (deterministic) arrival curve of $F_2$, we have $\bar{m}\geq m$ because of $a_2(m)\geq \lambda_2(m)$. Then $\gamma(n+\bar{m})\geq\gamma(n+m)$. Let $\gamma_1(n)=\gamma(n+\bar{m})$. From $\gamma_1(n)\geq \gamma(n+m)$, we have $a_1\bar{\otimes}\gamma_1(n)\geq a\bar{\otimes}\gamma(n+m)$. As $d(n+m)=d_1(n)$, there holds 
\begin{displaymath}
d_1(n) - a_1\bar{\otimes}\gamma_1(n) \leq d(n+m) - a\bar{\otimes}\gamma(n+m).
\end{displaymath} 
Thus, we conclude 
\begin{displaymath}
P\{d_1(n) - a_1\bar{\otimes}\gamma_1(n) > x\}\leq j(x)
\end{displaymath}
and complete the proof.
\end{proof}

\subsection{Discussion}
In this section, we have presented the five basic properties of stochastic network calculus under various traffic models and server models defined in the time-domain and introduced some simple applications.~For example, a GCRA-constrained flow has a deterministic arrival curve. If a flow's packet inter-arrival times are exponentially distributed, then this flow has a v.s.d stochastic arrival curve. The service process of an error-prone wireless link can be modeled by an i.d stochastic service curve. 

For each basic property, we investigated one combination of a specific traffic model and a specific server model. Particularly, we proved that the service guarantees and the output characterization hold for the combination of v.s.d stochastic arrival curve and i.d stochastic service curve. For the concatenation property, we investigated the case that all servers provide the constrained service curve to their input but did not specify the type of arrival curve. In order to prove the superposition property, we used the transformation between v.s.d stochastic service curve and v.b.c stochastic service curve. The leftover service characterization was only proved for the combination of deterministic arrival curve and i.d stochastic service curve.  

With the relationships and transformations among models established in Sec.~\ref{Section-Model}, these five properties may be directly or indirectly proved for other combinations of traffic models and server models. For example, it is easy to prove the service guarantees and output characterization for the combination of m.s.d stochastic arrival curve and i.d stochastic service curve.~Considering space limitation, these results are not included. However, to prove the concatenation property and superposition property for other server models and traffic models, it will require additional transformations among models.~For the leftover service characterization, we may need more constraints or transformations when proving it for other combinations of traffic models and server models. We leave these as our future work. 
    
\section{Conclusion}\label{Sec-Conclusion}
For stochastic service guarantee analysis, we introduced several {\em time-domain} models for traffic and service modeling. The essential idea of them is to base the model on cumulative packet inter-arrival time for traffic and on cumulative service time for service. Simple examples have been given to demonstrate the use of them.~Based on the proposed time-domain models, the five basic properties for stochastic network calculus were derived, with which, the results can be easily applied to both the single-node and the network cases. 

As Example 5 showed, we can directly obtain the service curve in the time-domain. With the result of service guarantees, the probabilistic delay bound and probablistic backlog bound can be readily obtained. We believe, the proposed time-domain models and derived results can be particularly useful for analyzing stochastic service guarantees in systems, where the behavior of a server involves some stochastic processes which can be directly characterized in the time-domain, while it is difficult to characterize such stochastic processes in the space-domain. Such systems include wireless links and multi-access networks where backoff schemes may be employed. 

In this paper, we only analyzed a simple case of wireless network to illustrate how to apply the proposed server model to characterize the service process of a wireless node. The future work is to investigate the performance of some typical contention-based multi-access networks including 
IEEE 802.11 networks.



\bibliographystyle{IEEEtran}
\bibliography{sigproc.bib}

\begin{thebibliography}{10}
\providecommand{\url}[1]{#1}
\csname url@samestyle\endcsname
\providecommand{\newblock}{\relax}
\providecommand{\bibinfo}[2]{#2}
\providecommand{\BIBentrySTDinterwordspacing}{\spaceskip=0pt\relax}
\providecommand{\BIBentryALTinterwordstretchfactor}{4}
\providecommand{\BIBentryALTinterwordspacing}{\spaceskip=\fontdimen2\font plus
\BIBentryALTinterwordstretchfactor\fontdimen3\font minus
  \fontdimen4\font\relax}
\providecommand{\BIBforeignlanguage}[2]{{%
\expandafter\ifx\csname l@#1\endcsname\relax
\typeout{** WARNING: IEEEtran.bst: No hyphenation pattern has been}%
\typeout{** loaded for the language `#1'. Using the pattern for}%
\typeout{** the default language instead.}%
\else
\language=\csname l@#1\endcsname
\fi
#2}}
\providecommand{\BIBdecl}{\relax}
\BIBdecl

\bibitem{Change:ExStoLinMax}
C.-S. Chang, ``On the exponentiality of stochastic linear systems under the
  max-plus algebra,'' \emph{IEEE Trans. Automatic Control}, vol.~41, no.~8, pp.
  1182--1188, Aug. 1996.

\bibitem{Fidler:End2end}
M.~Fidler, ``An end-to-end probabilistic network calculus with moment
  generating functions,'' in \emph{Proc. IEEE IWQoS 2006}, 2006, pp. 261--270.

\bibitem{Jiang:BasicStoNet}
Y.~Jiang, ``A basic stochastic network calculus,'' in \emph{Proc. ACM SIGCOMM
  2006}, 2006, pp. 123--134.

\bibitem{Jiang:book}
Y.~Jiang and Y.~Liu, \emph{Stochastic Network Calculus}.\hskip 1em plus 0.5em
  minus 0.4em\relax Springer, 2008.

\bibitem{Ferrari:Service}
D.~Ferrari, ``Client requirements for real-time communication services,''
  \emph{IEEE Commun. Magazine}, pp. 65--72, Nov. 1990.

\bibitem{Change:traffic}
C.-S. Chang, ``Stability, queue length and delay of deterministic and
  stochastic queueing networks,'' \emph{IEEE Trans. Auto. Control}, vol.~39,
  no.~5, pp. 913--931, May 1994.

\bibitem{Li:traffic}
C.~Li, A.~Burchard, and J.~Liebeherr, ``A network calculus with effective
  bandwidth,'' \emph{Technical Report, {CS-2003-20}, University of Virginia},
  Nov. 2003.

\bibitem{David:SBB}
D.~Starobinski and M.~Sidi, ``Stochastically bounded burstiness for
  communication networks,'' \emph{IEEE Trans. Information Theory}, vol.~46,
  no.~1, pp. 206--212, Jan. 2000.

\bibitem{Yaron:traffic}
O.~Yaron and M.~Sidi, ``Performance and stability of communication network via
  robust exponential bounds,'' \emph{IEEE/ACM Trans. Networking}, vol.~1,
  no.~3, pp. 372--385, June 1993.

\bibitem{Goyal:GR}
P.~Goyal, S.~S. Lam, and H.~M. Vin, ``Determining end-to-end delay bounds in
  heterogeneous networks,'' \emph{Multimedia System}, vol.~5, no.~3, pp.
  157--163, May 1997.

\bibitem{Cobb:Qua}
J.~A. Cobb, ``Preserving quality of service guarantees in spite of flow
  aggregation,'' \emph{IEEE/ACM Trans. Networking}, vol.~10, no.~1, pp. 43--53,
  Feb. 2002.

\bibitem{jiang:Time}
Y.~Jiang, ``Per-domain packet scale rate guarantee for expedited forwarding,''
  \emph{IEEE/ACM Trans. Networking}, vol.~14, no.~3, pp. 630--643, June 2006.

\bibitem{Boudec:calculus}
J.-Y. {Le Boudec} and P.~Thiran, \emph{Network calculus: A Theory of
  Deterministic Queueing Systems for the Internet}.\hskip 1em plus 0.5em minus
  0.4em\relax Springer, LNCS, 2004.

\bibitem{Jiang:TrafficModel}
Y.~Jiang and P.~J. Emstad, ``Analysis of stochastic service guarantees in
  communication networks: A traffic model,'' in \emph{Proc. 19th International
  Teletraffic Congress (ITC19)}, 2005.

\bibitem{Chang:MaxPlus}
C.-S. Chang and Y.~H. Lin, ``A general framework for deterministic service
  guarantees in telecommunication networks with variable length packets,''
  \emph{IEEE/ACM Trans. Automatic Control}, vol.~46, no.~2, pp. 210--221, Feb.
  2001.

\bibitem{ITU:GCRA}
{ITU-TSS Study Group 13}, ``Recommendation {I.371} traffic control and
  congestion control in {B-ISDN},'' 1995.

\bibitem{handbook:erlangdist}
\emph{NIST/SEMATECH e-Handbook of Statistical Methods}, 2006.

\bibitem{Shortle:MD1queue}
J.~F. Shortle and P.~H. Brill, ``Analytical distribution of waiting time in the
  {M/iD/1} queue,'' \emph{Queueing Systems}, vol.~50, no.~2, pp. 185--197,
  2005.

\end{thebibliography}
%



\end{document}